\IfFileExists{venue_preamble.tex}{%
  \def\VenueMode{1}
  \input{venue_preamble}
}{%
  \documentclass[11pt]{article}
  \usepackage[margin=1in]{geometry}
}

\usepackage[T1]{fontenc}
\usepackage{microtype}
\usepackage{amsmath}
\usepackage{amssymb}
\usepackage{amsthm}
\usepackage{aliascnt}
\usepackage{mathtools}
\usepackage{graphicx}
\usepackage{booktabs}
\usepackage{xcolor}
\usepackage{hyperref}

\ifdefined\VenueMode
\else
  \hypersetup{
    colorlinks=true,
    citecolor=blue,
    linkcolor=blue,
    urlcolor=blue
  }
\fi

\usepackage[nameinlink,capitalize,noabbrev]{cleveref}

\graphicspath{{figures/}}
\allowdisplaybreaks

\newcommand{\R}{\mathbb{R}}
\newcommand{\Q}{\mathbb{Q}}
\newcommand{\N}{\mathbb{N}}

\newcommand{\norm}[1]{\left\lVert #1\right\rVert}
\newcommand{\abs}[1]{\left\lvert #1\right\rvert}
\newcommand{\set}[1]{\left\{#1\right\}}

\newcommand{\ReLU}{\operatorname{ReLU}}
\newcommand{\Lip}{\operatorname{Lip}}
\newcommand{\OPT}{\operatorname{OPT}}
\newcommand{\poly}{\operatorname{poly}}
\newcommand{\sign}{\operatorname{sign}}
\newcommand{\conv}{\operatorname{conv}}
\newcommand{\grad}{\nabla}

\theoremstyle{plain}
\newtheorem{theorem}{Theorem}[section]
\newaliascnt{lemma}{theorem}
\newtheorem{lemma}[lemma]{Lemma}
\aliascntresetthe{lemma}
\newaliascnt{proposition}{theorem}
\newtheorem{proposition}[proposition]{Proposition}
\aliascntresetthe{proposition}
\newaliascnt{corollary}{theorem}
\newtheorem{corollary}[corollary]{Corollary}
\aliascntresetthe{corollary}
\newaliascnt{conjecture}{theorem}

\aliascntresetthe{conjecture}

\theoremstyle{definition}
\newaliascnt{definition}{theorem}
\newtheorem{definition}[definition]{Definition}
\aliascntresetthe{definition}
\newaliascnt{assumption}{theorem}

\aliascntresetthe{assumption}
\newaliascnt{example}{theorem}

\aliascntresetthe{example}
\newaliascnt{problem}{theorem}
\newtheorem{problem}[problem]{Problem}
\aliascntresetthe{problem}
\newaliascnt{openproblem}{theorem}
\newtheorem{openproblem}[openproblem]{Open Problem}
\aliascntresetthe{openproblem}

\crefname{lemma}{Lemma}{Lemmas}
\Crefname{lemma}{Lemma}{Lemmas}
\crefname{proposition}{Proposition}{Propositions}
\Crefname{proposition}{Proposition}{Propositions}
\crefname{corollary}{Corollary}{Corollaries}
\Crefname{corollary}{Corollary}{Corollaries}
\crefname{definition}{Definition}{Definitions}
\Crefname{definition}{Definition}{Definitions}
\crefname{assumption}{Assumption}{Assumptions}
\Crefname{assumption}{Assumption}{Assumptions}
\crefname{example}{Example}{Examples}
\Crefname{example}{Example}{Examples}
\crefname{problem}{Problem}{Problems}
\Crefname{problem}{Problem}{Problems}
\crefname{openproblem}{Open Problem}{Open Problems}
\Crefname{openproblem}{Open Problem}{Open Problems}
\crefname{remark}{Remark}{Remarks}
\Crefname{remark}{Remark}{Remarks}

\theoremstyle{remark}
\newaliascnt{remark}{theorem}

\aliascntresetthe{remark}

\newcommand{\PaperTitle}{$\ell_p$-Norm Maximization over Zonotopes Is W[1]-Hard}
\newcommand{\PaperAuthors}{Yang Cao \and Haoran Qi \and Hanzhi Wang}
\newcommand{\PaperDate}{}

\newif\ifPaperPrintReferences
\PaperPrintReferencestrue

\title{\PaperTitle}
\author{\PaperAuthors}
\date{\PaperDate}

\providecommand{\PaperMakeTitle}{\maketitle}
\providecommand{\PaperPrintBibliography}{%
  \ifPaperPrintReferences
    \bibliographystyle{alpha}
    \bibliography{references}
  \fi}
\providecommand{\PaperEndMatter}{}

\begin{document}

\PaperMakeTitle

\begin{abstract}
We study $\ell_p$-norm maximization over zonotopes given by rational generators, with input length $L$. For fixed $p=a/b>1$, the exact Turing baseline runs in $n^{O(d)}b^{O(d)}\poly(L)$ time, but fixed-parameter tractability in the ambient dimension $d$ was open \cite{fgh+25}. We prove W[1]-hardness and, under the Exponential Time Hypothesis (ETH), exclude $\rho_p(d)L^{o(d)}$ time, even for $5$-sparse generators, by encoding binary CSP constraints with normalized positive cap generators. We also give a deterministic $(1-\varepsilon)$-approximation with $\varepsilon^{-(d-1)/2}$ dependence and, among algorithms with fixed-degree polynomial dependence on $L$, rule out $(1/\varepsilon)^{o(d)}$ dependence under ETH. Support-function duality transfers the results to positive-output two-layer ReLU networks.

\end{abstract}

\section{Introduction}

Given rational generators $a_1,\ldots,a_n\in\Q^d$, let
\begin{align*}
Z(A)
&\coloneqq A[0,1]^n
=\sum_{i=1}^n[0,a_i]
\end{align*}
be their zonotope.  We study the farthest point of $Z(A)$ in a fixed $\ell_p$ norm.

\begin{problem}[Fixed-$\ell_p$ norm maximization over zonotopes]
\label[problem]{prob:lvs}
Fix a rational constant $p\in(1,\infty)$.  Given a rational matrix $A\in\Q^{d\times n}$ and a rational threshold $B>0$, decide whether
\begin{align*}
\max_{z\in Z(A)}\norm{z}_p^p
&\geq B.
\end{align*}
The parameter is the ambient dimension $d$, and $L$ denotes the total binary encoding length of the rational data in the instance.
\end{problem}

Since the norm is convex, its maximum over $Z(A)=A[0,1]^n$ is attained at the image of a cube vertex.  Thus the associated optimization value is
\begin{align*}
R_p(A)
&\coloneqq\max_{z\in Z(A)}\norm{z}_p
=\max_{x\in\set{0,1}^n}\norm{Ax}_p.
\end{align*}
The Boolean formulation is also known as fixed-$\ell_p$ Longest Vector Sum.  When $p=2$, $R_p(A)^2$ is the positive-semidefinite binary quadratic objective $\max_x x^\top A^\top Ax$, whose matrix $A^\top A$ has rank at most $d$.

The central question is whether the problem is fixed-parameter tractable in the ambient dimension.
\begin{openproblem}[\cite{fgh+25}]
Is $L_p$-Maximization over $d$-Zonotopes fixed-parameter tractable with respect to $d$ for $p>1$?
\end{openproblem}

The classical vertex-enumeration algorithm places the problem in XP.
\begin{theorem}[Theorem 2.1 of \cite{fgh+26}]
\label{thm:known-xp}
With $n$ generators in dimension $d$ and input bit-length $N$, Zonotope Containment and $L_p$-Max on Zonotopes are solvable in $\mathcal O(n^{d-1}\poly(N))$ time.
\end{theorem}
For a nonintegral rational exponent, an ordinary Turing implementation must also account for exact comparison of sums of algebraic numbers.  \Cref{prop:turing-language} gives the binary decision language explicitly, and \cref{prop:turing-xp} proves that for $p=a/b$ each candidate comparison takes $b^{O(d)}\poly(L)$ bit operations.  The resulting exact Turing running time is $n^{O(d)}b^{O(d)}\poly(L)$.
The exponent $d-1$ in \cref{thm:known-xp} is exactly the obstacle: fixed-parameter tractability would require a running time $f(d)\poly(L)$ whose polynomial degree is independent of $d$.  The question remained open after the parameterized classification of adjacent zonotope and neural-network problems in \cite{fgh+26}.  The same open-problem note separately asked whether zonotope containment is fixed-parameter tractable; that question was resolved in \cite{fgh+26}, whereas the norm-maximization question above was retained.  Ordinary NP-hardness does not settle this question, since an NP-hard problem can still be fixed-parameter tractable in $d$.  The endpoint cases are also inconclusive: $\ell_1$ maximization over generator zonotopes is fixed-parameter tractable, whereas $\ell_\infty$ maximization is polynomial-time solvable \cite{bp07,fgh+26}.

For halfspace-presented symmetric polytopes, \cite{kkw15} proved the corresponding W[1]-hardness by placing labels on a product of planar $\ell_p$ spheres and intersecting the product with graph-dependent strips.  The strips delete forbidden label pairs, but intersection with a strip does not preserve a Minkowski sum of segments.  A reduction for generator zonotopes must instead encode binary relations using only additional generators.

\begin{theorem}[Generator-zonotope hardness, informal version of \cref{thm:exact-hardness}]
\label{thm:main-informal}
For every fixed rational $p\in(1,\infty)$, \cref{prob:lvs} is W[1]-hard parameterized by $d$, even when every column of $A$ is $5$-sparse.  A $k$-variable binary CSP maps to dimension $2k+1$, and under the Exponential Time Hypothesis (ETH) there is no algorithm with running time $\rho_p(d)L^{o(d)}$ for any computable function $\rho_p$.
\end{theorem}

Together with \cref{prop:turing-xp}, the ETH lower bound shows that the linear dependence of the input-size exponent on $d$ is optimal up to constant factors.

A size-restricted specialization is sharper: \cref{cor:small-instance-lower} rules out $2^{o(d\log d)}\poly(L)$ time under ETH even when $n,L=d^{O(1)}$, while \cref{prop:turing-xp} runs in $2^{O(d\log d)}$ time on the same family.

The same generator matrix defines the positive homogeneous network
\begin{align*}
f_A(u)
&\coloneqq\sum_{i=1}^n\ReLU(a_i^\top u).
\end{align*}
For dual exponents $p$ and $q$, its $\ell_q$-Lipschitz constant is $R_p(A)$.  The hardness therefore holds with every output weight equal to $+1$, every bias zero, and every hidden weight $5$-sparse.  This is the positive-output case isolated in \cite{fgh+25,fgh+26}; the signed-output reductions for general two-layer networks do not preserve a single generator zonotope.

The cap gap and a deterministic upper bound give the following precision frontier.

\begin{theorem}[Precision frontier, informal versions of \cref{thm:precision-upper,thm:precision-lower}]
\label{thm:precision-main-informal}
For every fixed rational $p\in(1,\infty)$ and rational $\varepsilon\in(0,1/2)$, there is a deterministic algorithm that returns a subset $\widehat x\in\set{0,1}^n$ satisfying
\begin{align*}
\norm{A\widehat x}_p
&\geq(1-\varepsilon)R_p(A)
\end{align*}
in $f_p(d)\varepsilon^{-(d-1)/2}\poly(L,\log(1/\varepsilon))$ bit operations and polynomial working space.  Conversely, unless FPT$=$W[1], no algorithm with running time $f_p(d)\poly(L,1/\varepsilon)$ guarantees a $(1-\varepsilon)$ ratio.  Under ETH, no such approximation has running time $\rho_p(d)(1/\varepsilon)^{h(d)}L^c$ for any computable $\rho_p$, constant $c$, and function $h(d)=o(d)$.  The lower bounds hold even for $5$-sparse columns.
\end{theorem}

Among algorithms whose dependence on $L$ has fixed polynomial degree, the upper bound uses an $O(d)$ exponent of $1/\varepsilon$, while ETH rules out every $o(d)$ exponent.  The real-arithmetic core of the upper bound is implicit in classical convex-body approximation; our formulation makes the returned subset and the rational Turing cost explicit and couples it to the generator-side lower bound.

\section{Model and Preliminaries}

Fix a rational constant $p=a/b\in(1,\infty)$ in lowest terms; $p$ is not part of the input.  Let
\begin{align*}
F_p(z)
&\coloneqq\norm{z}_p^p
=\sum_{j=1}^d\abs{z_j}^p,
&
J_p(z)_j
&\coloneqq\sign(z_j)\abs{z_j}^{p-1}.
\end{align*}
Then $\grad F_p(z)=pJ_p(z)$.  Maximizing $F_p$ and maximizing the norm have the same solutions.

For $A\in\Q^{d\times n}$, write
\begin{align*}
\OPT_p(A)
&\coloneqq\max_{x\in\set{0,1}^n}F_p(Ax).
\end{align*}
Thus \cref{prob:lvs} asks whether $\OPT_p(A)\geq B$.

\begin{proposition}[Exact Turing formulation]
\label{prop:turing-language}
Under the standard binary encoding of integers and reduced rational numerator--denominator pairs, define
\begin{align*}
\mathrm{LVS}_p
&\coloneqq
\set{
\left\langle d,n,A,B\right\rangle:
d,n\in\N,
A\in\Q^{d\times n},
B\in\Q_{>0},
\exists x\in\set{0,1}^n\quad
F_p(Ax)\geq B
}.
\end{align*}
Malformed encodings are excluded from the language.  For every fixed rational $p=a/b>1$, $\mathrm{LVS}_p$ is a decidable language in the ordinary Turing model.
\end{proposition}
\begin{proof}
For a valid input, introduce real variables $x_i,z_j,r_j,y_j$.  After clearing rational denominators, membership is equivalent to the existence of a real solution to
\begin{align*}
x_i(1-x_i)
&=0
&&\text{for }i\in[n],
&
z_j
&=\sum_{i=1}^n A_{ji}x_i
&&\text{for }j\in[d],
\\
r_j
&\geq0,
&
r_j^2
&=z_j^2
&&\text{for }j\in[d],
\\
y_j
&\geq0,
&
y_j^b
&=r_j^a
&&\text{for }j\in[d],
&
\sum_{j=1}^d y_j
&\geq B.
\end{align*}
The first line forces a Boolean vector and its rational subset sum, the second forces $r_j=\abs{z_j}$, and the third forces $y_j=\abs{z_j}^{a/b}$.

The displayed system is an existential formula over the reals with integer coefficients and is therefore decidable by real quantifier elimination \cite[pp.~465--521]{bpr03}.
\end{proof}

\begin{proposition}[Exact comparison and Turing-XP]
\label{prop:turing-xp}
Let $z,z'\in\Q^d$ have total encoding length at most $L'$.  The sign of $F_p(z)-F_p(z')$, including equality, can be determined in $b^{O(d)}\poly(L')$ bit operations.  Consequently, $\mathrm{LVS}_p$ can be decided and a maximizing subset can be returned in
\begin{align*}
n^{O(d)}b^{O(d)}\poly(L)
\end{align*}
bit operations.
\end{proposition}
\begin{proof}
For each nonzero coordinate, express $\abs{z_j}=r_j/s_j$ in lowest terms.  The number
\begin{align*}
\alpha_j
&\coloneqq\abs{z_j}^{a/b}
\end{align*}
is the distinguished nonnegative root of
\begin{align*}
s_j^aX^b-r_j^a.
\end{align*}
The polynomial has degree at most $b$ and coefficient bit length $O(aL')$; the same description applies to the terms from $z'$.  Give the at most $t\leq2d$ nonzero terms their signs in $F_p(z)-F_p(z')$, replacing $X$ by $-X$ in the defining polynomial of a negative term.

Repeated Sylvester resultants construct a nonzero polynomial $R\in\mathbb Z[X]$ that annihilates their signed sum.  Namely, if $P$ and $Q$ annihilate two partial sums, then
\begin{align*}
\operatorname{Res}_Y\bigl(P(Y),Q(X-Y)\bigr)
\end{align*}
annihilates their sum.  The Sylvester determinant bound shows inductively that $R$ has degree $D\leq b^t$ and coefficient bit length $b^{O(t)}\poly(L')$, and it can be constructed within the same bit bound.

Factor $R$ as
\begin{align*}
R(X)
&=X^m\left(c_m+c_{m+1}X+\cdots+c_DX^{D-m}\right),
&
c_m
&\neq0.
\end{align*}
Unless the parenthesized polynomial is constant, every nonzero root $\zeta$ of $R$ satisfies the rational lower bound
\begin{align*}
\abs{\zeta}
&\geq
\Delta
\coloneqq
\frac{\abs{c_m}}{\abs{c_m}+\sum_{i>m}\abs{c_i}}.
\end{align*}
This follows by applying the triangle inequality to the parenthesized equation when $\abs{\zeta}<1$; the case $\abs{\zeta}\geq1$ is immediate.  The encoding length of $\Delta$ is $b^{O(d)}\poly(L')$.  Each $\alpha_j$ is the unique nonnegative root of a fixed-degree binomial, so rational bisection encloses every signed term tightly enough that the interval sum has width less than $\Delta/2$.  If the sum interval avoids zero, its side gives the sign.  If it contains zero, the represented sum has magnitude below $\Delta$ and must equal zero.  If $R$ is a monomial, the represented sum is zero directly.  This proves exact comparison in $b^{O(d)}\poly(L')$ time.

For the algorithmic consequence, delete zero columns and enumerate the full-dimensional cells of the rational central arrangement $a_i^\top y=0$ used in \cref{thm:known-xp}.  A rational representative of each cell exposes the subset sum
\begin{align*}
z(y)
&=\sum_{i:y^\top a_i>0}a_i.
\end{align*}
Conversely, every zonotope vertex is exposed by an interior point of its full-dimensional normal cone and therefore occurs among these candidates.  There are $n^{O(d)}$ cells, they can be enumerated by exact rational arithmetic in $n^{O(d)}\poly(L)$ time, and every candidate has encoding length $\poly(L)$.  Applying the comparison procedure to the candidates retains a maximizer.  The same procedure compares its value with $B$ by including $-B$ as one additional rational algebraic term of degree one, which gives the claimed exact Turing decision bound.
\end{proof}

The zonotope generated by $A$ is
\begin{align*}
Z(A)
&\coloneqq A[0,1]^n
=\sum_{i=1}^n[0,a_i].
\end{align*}
Convexity implies that a maximum of $F_p$ over $Z(A)$ is attained at a vertex, hence at a subset sum.  For a direction $y\in\R^d$, define
\begin{align*}
S(y)
&\coloneqq\set{i:y^\top a_i\geq0},
&
z(y)
&\coloneqq\sum_{i\in S(y)}a_i,
&
h_A(y)
&\coloneqq y^\top z(y)=\sum_i\max\set{0,y^\top a_i}.
\end{align*}
Ties may be included or excluded without changing the support value.

We use the standard parameterized assumptions FPT$\neq$W[1] and ETH.  The source problem is $k$-Multicolored Clique: the input graph has $k$ color classes of size at most $N$, and the task is to select one vertex from each class so that all selected vertices are pairwise adjacent.

\begin{theorem}[Multicolored Clique lower bounds \cite{pie03,chk+06}]
\label{thm:clique-lower}
$k$-Multicolored Clique is W[1]-hard in $k$.  Under ETH it has no $f(k)N^{o(k)}$ algorithm.
\end{theorem}

For a binary CSP, let the variables be $1,\ldots,k$, let $D$ be a common finite domain, and let each of the $K$ constraints specify a binary relation on a pair of distinct variables.  The input explicitly lists the set of accepted tuples for each constraint.  Multicolored Clique is the special case with one variable per color class and one adjacency constraint per pair of colors.

\section{The Generator-Preserving Reduction}
\label{sec:hardness}

We prove the main theorem through a quantitative reduction from binary CSP.  A variable assignment is represented by a vertex of a rational planar zonogon.  For each accepted pair of labels, a positive cap generator has positive first-order score only at that pair and negative score at every other pair.  A common perturbation scale then makes the objective distinguish assignments by the number of satisfied constraints.  We first establish the sign rule that controls generator selection, then construct the labels and caps, assemble the CSP encoder, and finally specialize it to Multicolored Clique.

\begin{lemma}[Global self-consistency]
\label{lem:self-consistency}
Let $\Phi:\R^d\to\R$ be differentiable and strictly convex.  Suppose that $z=\sum_{i\in S}a_i$ globally maximizes $\Phi$ over all subset sums and that every $a_i$ is nonzero.  Then
\begin{align*}
i\in S
&\quad\Longleftrightarrow\quad
\grad\Phi(z)^\top a_i>0.
\end{align*}
In particular, if the generator list contains $g$ and $-g$, exactly one of them is selected at every global optimum.
\end{lemma}
\begin{proof}
If $i\in S$, global optimality gives $\Phi(z)\geq\Phi(z-a_i)$.  Strict convexity and $a_i\neq0$ give
\begin{align*}
\Phi(z-a_i)
&>\Phi(z)-\grad\Phi(z)^\top a_i,
\end{align*}
so $\grad\Phi(z)^\top a_i>0$.  If $i\notin S$, then $\Phi(z)\geq\Phi(z+a_i)$ and
\begin{align*}
\Phi(z+a_i)
&>\Phi(z)+\grad\Phi(z)^\top a_i,
\end{align*}
so $\grad\Phi(z)^\top a_i<0$.  The scores of $g$ and $-g$ have opposite signs.
\end{proof}

The reduction next uses a planar label set with two properties: every label is a vertex of a rational zonogon, and the norming functional of one label loses an inverse-polynomial amount on every other label.  We work on a compact arc away from the coordinate axes so that rational approximation and second-order perturbation are uniform for every fixed $p>1$.

\begin{lemma}[Rational curved labels]
\label{lem:rational-labels}
For every fixed rational $p\in(1,\infty)$ there is a rational constant $\underline\mu_p>0$ with the following property.  Given $m\geq2$ and a positive rational $\eta$, one can construct in polynomial time a rational number
\begin{align*}
0<\Gamma
&\leq\min\set{\frac{\underline\mu_p}{1024(m-1)^2},\frac{1}{100}},
\end{align*}
a centrally symmetric set $V\coloneqq\set{a_1,\ldots,a_m,-a_1,\ldots,-a_m}\subset\Q^2$, and nonzero rational vectors $g_1,\ldots,g_m$ such that
\begin{itemize}
\item every coordinate of every label has absolute value at least $3/16$;
\item $\abs{F_p(v)-1}\leq\eta$ for every $v\in V$;
\item $J_p(v)^\top w\leq1-8\Gamma$ for distinct $v,w\in V$; and
\item $\conv(V)=\sum_{j=1}^m[-g_j,g_j]$.
\end{itemize}
The total encoding length is polynomial in $m$ and the encoding length of $\eta$.
\end{lemma}
\begin{proof}
Let $I\coloneqq[1/4,1/2]$ and define
\begin{align*}
h(t)
&\coloneqq(1-t^p)^{1/p},
&
u(t)
&\coloneqq(t,h(t)).
\end{align*}
Then $F_p(u(t))=1$, both coordinates are bounded away from zero, and
\begin{align*}
h'(t)
&=-t^{p-1}h(t)^{1-p},
&
h''(t)
&=-(p-1)t^{p-2}h(t)^{1-2p}.
\end{align*}
For fixed $s\in I$, put $H_s(t)\coloneqq J_p(u(s))^\top u(t)$.  We have $H_s(s)=1$, $H_s'(s)=0$, and $H_s''(t)\leq-\mu_p$ on $I$ for a constant $\mu_p>0$ depending only on $p$.  Since $p$ is fixed and rational, choose once and for all a positive rational lower bound $\underline\mu_p\leq\mu_p$.  Taylor's theorem gives
\begin{align*}
J_p(u(s))^\top u(t)
&\leq1-\frac{\underline\mu_p}{2}(s-t)^2.
\end{align*}

Take $t_j\coloneqq1/4+(j-1)/(4(m-1))$ and the labels $\set{u(t_j),-u(t_j):j\in[m]}$.  Distinct labels have support score at most $1-\gamma_0$, where
\begin{align*}
\gamma_0
&\coloneqq\frac{\underline\mu_p}{32(m-1)^2}.
\end{align*}
The antipodal cases follow from the oddness of $J_p$ and the positivity of the first-quadrant coordinates.

Let
\begin{align*}
T_\Gamma
&\coloneqq\min\set{\frac{\gamma_0}{32},\frac{1}{100}},
\end{align*}
and choose $\Gamma$ as the largest number of the form $2^{-r}$, $r\in\mathbb Z_{\geq0}$, that is at most $T_\Gamma$.  Then
\begin{align*}
\frac{T_\Gamma}{2}
&<\Gamma\leq T_\Gamma,
\end{align*}
so $\Gamma^{-1}=O_p(m^2)$ and $\Gamma$ has $O_p(\log m)$ encoding length.

Keep the first coordinate $t_j$ rational and replace the second coordinate by a dyadic approximation.  On a fixed compact neighborhood of the arc and its antipode, the maps $F_p$, $J_p$, and $(v,w)\mapsto J_p(v)^\top w$ are Lipschitz, with constants depending only on $p$.  Certified bisection for the fixed-degree root $(1-t_j^p)^{1/p}$ therefore finds polynomial-bit dyadics that preserve the coordinate lower bound, make the norm error at most $\eta$, and perturb every support score by at most $\gamma_0/4$.  Choosing $\Gamma\leq\gamma_0/32$ gives the first three conclusions.

The exact functional $J_p(u(t_j))$ continues to expose the corresponding rational perturbation, so every point of $V$ is a vertex of the centrally symmetric rational polygon $P\coloneqq\conv(V)$.  Every centrally symmetric convex polygon is a zonotope.  Compute the cyclic order by exact rational planar convex hull.  If one half is $v_1,\ldots,v_m$, its rational generators are
\begin{align*}
g_j
&\coloneqq\frac{v_{j+1}-v_j}{2} &&\text{for }j<m,
&
g_m
&\coloneqq\frac{-v_1-v_m}{2}.
\end{align*}
Thus $P=\sum_j[-g_j,g_j]$, and all operations after the dyadic approximation are exact rational operations.
\end{proof}

Fix two variable blocks and labels $r,s\in V$.  An anchored scaffold assignment has the form $q=(\ldots,r,\ldots,s,\ldots,1)$.  A cap aimed at $(r,s)$ has raw entries $r$ and $s$ in those blocks and a negative entry in the anchor coordinate.

\begin{lemma}[Normalized positive caps]
\label{lem:cap-normalization}
Assume the labels of \cref{lem:rational-labels} with $\eta\leq\Gamma/8$, and let $0<\xi\leq1/100$ be rational.  For every ordered target pair $(r,s)$ there is a rational $5$-sparse vector $c_{rs}$ such that
\begin{align*}
1-\xi
&\leq\grad F_p(q_{rs})^\top c_{rs}\leq1+\xi,
\end{align*}
where $q_{rs}$ is the anchored scaffold point containing $r,s$, while for every scaffold point $q_{xy}$ with $(x,y)\neq(r,s)$,
\begin{align*}
\grad F_p(q_{xy})^\top c_{rs}
&<-\frac45.
\end{align*}
Moreover $c_{rs}$ can be computed in polynomial time and satisfies $\norm{c_{rs}}_2\leq C_0$ for a rational bound $C_0\leq c_p\Gamma^{-1}$ that is independent of the final label precision.
\end{lemma}
\begin{proof}
Set $\theta\coloneqq2-4\Gamma$ and let $b_{rs}$ contain $r$, $s$, and $-\theta$ in the two label blocks and anchor coordinate.  At the target,
\begin{align*}
\frac1p\grad F_p(q_{rs})^\top b_{rs}
&=J_p(r)^\top r+J_p(s)^\top s-\theta
\in\left[\frac{15\Gamma}{4},\frac{17\Gamma}{4}\right].
\end{align*}
If at least one label is wrong, one diagonal score is at most $1+\eta$ and one off-diagonal score is at most $1-8\Gamma$, so
\begin{align*}
\frac1p\grad F_p(q_{xy})^\top b_{rs}
&\leq\eta-4\Gamma
\leq-\frac{31\Gamma}{8}.
\end{align*}

Let $D\coloneqq\grad F_p(q_{rs})^\top b_{rs}$.  Because $D$ involves only five coordinates and $p$ is fixed rational, certified fixed-degree root bisection produces a rational interval around $D$ of relative width at most $\xi/4$.  Choose a positive rational $\lambda$ from the corresponding reciprocal interval so that $1-\xi\leq\lambda D\leq1+\xi$, and set $c_{rs}\coloneqq\lambda b_{rs}$.  Since $D\leq17p\Gamma/4$, the wrong-label bound is at most
\begin{align*}
-\frac{31p\Gamma}{8}\frac{1-\xi}{17p\Gamma/4}
&=-\frac{31(1-\xi)}{34}
<-\frac45.
\end{align*}
Finally $\lambda=O_p(\Gamma^{-1})$, while the raw cap has bounded coordinates.  A rational bound $C_0\leq c_p\Gamma^{-1}$ can therefore be fixed uniformly before the last dyadic label precision is chosen.
\end{proof}

We now assemble the scaffold and caps.  The quantitative statement will also supply the approximation gap used in \cref{sec:approximation}.

\begin{theorem}[Binary-CSP encoder]
\label{thm:csp-encoder}
Fix a rational $p\in(1,\infty)$.  Given an explicit binary CSP with $k$ variables, a common domain of size $2m$ with $m\geq2$, $K$ constraints on pairs of distinct variables, and $M$ listed accepted tuples, one can construct in polynomial time a rational generator matrix $A$ in dimension $d=2k+1$ with the following properties.

Every generator is at most $5$-sparse.  There are positive rationals $\tau,\eta,\xi$ such that every assignment satisfying $s$ constraints has a canonical subset whose value lies in
\begin{align}
k+1-k\eta+\tau s(1-\xi)-\frac{\tau}{20}
&\leq F_p(z)\leq
k+1+k\eta+\tau s(1+\xi)+\frac{\tau}{20}.
\label{eq:csp-interval}
\end{align}
Every global maximizer is the canonical subset of its encoded assignment, so it satisfies the same bounds with $s$ equal to the number of constraints satisfied by that assignment.  The intervals for consecutive values of $s$ are separated by more than $4\tau/5$.  All parameters and generators have polynomial encoding length.
\end{theorem}
\begin{proof}
The construction in the proof of \cref{lem:rational-labels} first fixes $\Gamma$ from $p$ and $m$, independently of the later accuracy $\eta$.  Fix this $\Gamma$, retain the notation $T_\Gamma$ for its dyadic upper threshold, and fix the uniform bound $C_0\leq c_p\Gamma^{-1}$ from \cref{lem:cap-normalization}; the bound is valid for every later admissible rationalization of the labels.  Let
\begin{align*}
\xi
&\coloneqq\frac{1}{100(K+1)},
&
\overline M
&\coloneqq\max\set{1,M},
&
\overline k
&\coloneqq\max\set{1,k}.
\end{align*}
Let $G_p\geq1$ bound every coordinate of $\grad F_p$ on $[-3,3]$, and let $B_p\geq1$ bound the Euclidean operator norm of the diagonal Hessian of $F_p$ on $1/8\leq\abs{x_j}\leq3$.  These rational constants depend only on $p$.  Define
\begin{align*}
T_\tau
&\coloneqq\min\set{
1,
\frac{1}{16\overline M C_0},
\frac{1}{4dG_p\overline M C_0},
\frac{1}{4B_p\overline M C_0^2},
\frac{1}{10B_p\overline M^2 C_0^2}
},
\end{align*}
and choose $\tau$ as the largest number of the form $2^{-r}$, $r\in\mathbb Z_{\geq0}$, that is at most $T_\tau$.  Thus $T_\tau/2<\tau\leq T_\tau$.  Next define
\begin{align*}
T_\eta
&\coloneqq\min\set{\frac{\Gamma}{8},\frac{\tau}{100\overline k}},
\end{align*}
and choose $\eta$ as the largest dyadic no larger than $T_\eta$, so $T_\eta/2<\eta\leq T_\eta$.

Only now construct the labels at accuracy $\eta$ and normalize the caps.  Use one orthogonal two-dimensional block for each variable.  In every block insert both $g_j$ and $-g_j$ for every generator of the label zonogon, insert the anchor $e_{2k+1}$, and, for each listed accepted tuple of a constraint, insert the corresponding cap from \cref{lem:cap-normalization} scaled by $\tau$.  The dependency order is acyclic because $C_0$ is uniform over every later admissible rationalization.

Let $z\coloneqq q+w$ be a global optimum, where $q$ is its scaffold sum together with the possible anchor and $w$ is the sum of its selected scaled caps.  Applying \cref{lem:self-consistency} to every pair $g_j,-g_j$ shows that exactly one member of each pair is selected.  The chosen signs in one block are the signs of one linear functional, namely the corresponding block of $\grad F_p(z)$.  Their sum is therefore a vertex of the zonogon, hence one label in $V$.  Thus $q$ encodes one label per variable even before the selected caps are identified.

Put $H\coloneqq\tau M C_0$.  Then $\norm{w}_2\leq H\leq1/16$.  If the anchor is absent, the bounded-gradient estimate gives
\begin{align*}
F_p(q+w)
&\leq k(1+\eta)+\sqrt dG_pH
\leq k(1+\eta)+\frac14.
\end{align*}
An anchored scaffold vertex with no caps is feasible and has value at least $k(1-\eta)+1$.  Since $k\eta\leq\tau/100\leq1/100$, the latter is larger, so every global optimum contains the anchor.

Every coordinate of the anchored $q$ now has absolute value at least $3/16$ in a label block and equals one in the anchor coordinate.  Along the segment $q+tw$, all coordinates stay in the region $1/8\leq\abs{x_j}\leq3$.  For every unscaled cap $c$, the Hessian bound and the parameter choice give
\begin{align*}
\abs{(\grad F_p(q+w)-\grad F_p(q))^\top c}
&\leq B_pHC_0
\leq\frac14.
\end{align*}
The target score from \cref{lem:cap-normalization} remains positive, and every wrong-label score remains negative.  Another application of \cref{lem:self-consistency} shows that a cap is selected at the global optimum exactly when its accepted tuple agrees with the encoded assignment.  Hence precisely one cap of a constraint is selected if and only if that constraint is satisfied.

For an assignment satisfying $s$ constraints, take its scaffold representation, the anchor, and the $s$ matching caps.  Taylor's theorem on the same segment gives
\begin{align*}
\abs{F_p(q+w)-F_p(q)-\grad F_p(q)^\top w}
&\leq\frac{B_pH^2}{2}
\leq\frac{\tau}{20}.
\end{align*}
The scaffold value lies in $[k+1-k\eta,k+1+k\eta]$, and the first-order contribution of the matching caps lies in $[\tau s(1-\xi),\tau s(1+\xi)]$.  This proves \eqref{eq:csp-interval}.  The preceding selection argument shows that every global maximizer has exactly this canonical form for its encoded assignment, so the same interval applies to it.  The difference between the lower interval for $s$ and the upper interval for $s-1$ is at least
\begin{align*}
\tau\left(1-(2K-1)\xi-\frac1{10}\right)-2k\eta
&>\frac{4\tau}{5}.
\end{align*}

The dimension is $2k+1$.  Scaffold columns are $2$-sparse, the anchor is $1$-sparse, and a cap touches two planar blocks and the anchor.  The maximal-dyadic choices give
\begin{align*}
\Gamma^{-1}
&<2T_\Gamma^{-1},
&
\tau^{-1}
&<2T_\tau^{-1},
&
\eta^{-1}
&<2T_\eta^{-1}.
\end{align*}
Here $T_\Gamma^{-1}=O_p(m^2)$, $C_0=O_p(\Gamma^{-1})$, and the displayed definitions make $T_\tau^{-1}$ and $T_\eta^{-1}$ polynomial in $m,M,K,k$.  Hence every parameter, label, normalizer, and generator has polynomial encoding length.
\end{proof}

The encoder immediately yields \cref{thm:main-informal}.

\begin{theorem}[Generator-zonotope hardness]
\label{thm:exact-hardness}
For every fixed rational $p\in(1,\infty)$, \cref{prob:lvs} is W[1]-hard parameterized by $d$, even when every column of $A$ is $5$-sparse.  Under ETH, it admits no algorithm with running time $\rho_p(d)L^{o(d)}$ for any computable function $\rho_p$.
\end{theorem}
\begin{proof}
Pad the color classes with isolated dummy labels to a common even size at least four, and apply \cref{thm:csp-encoder} with one variable per color class and one adjacency constraint per pair of classes.  The padding does not create a clique.  If the graph has a multicolored clique, some assignment satisfies all $K=\binom{k}{2}$ constraints.  Otherwise every assignment satisfies at most $K-1$.  The objective intervals in \cref{thm:csp-encoder} are separated by more than $4\tau/5$, so the midpoint between the lower endpoint for $s=K$ and the upper endpoint for $s=K-1$ is a valid rational threshold.

The output dimension is $2k+1$, and its bit length is polynomial in $k$ and $N$.  W[1]-hardness follows from \cref{thm:clique-lower}.  An algorithm running in $\rho_p(d)L^{o(d)}$ time would yield an $f_p(k)N^{o(k)}$ algorithm for Multicolored Clique, contradicting ETH.
\end{proof}

The same encoder yields a tight total-time bound on instances whose generator list and bit length are polynomial in the dimension.

\begin{corollary}[Polynomial-size hard instances]
\label{cor:small-instance-lower}
For every fixed rational $p\in(1,\infty)$, under ETH there is no algorithm for \cref{prob:lvs} with running time $2^{o(d\log d)}\poly(L)$, even on instances with $n,L=d^{O(1)}$ and $5$-sparse columns.  On this restricted family, the known $\mathcal O(n^{d-1}\poly(L))$ upper bound is $2^{O(d\log d)}$.
\end{corollary}
\begin{proof}
Theorem~2.4 of \cite{lms18} states that, under ETH, $k\times k$ Clique has no $2^{o(k\log k)}$-time algorithm.  Its graph has vertex set $[k]\times[k]$, and the task is to select one vertex from each row so that the selected vertices form a clique.

Create one CSP variable for each row, with its value specifying the selected column.  Let $q$ be the smallest even integer at least $\max\set{k,4}$, and pad each row with isolated dummy labels, so $q=O(k)$.  For each pair of rows, the binary relation consists of their adjacent column pairs.  Thus
\begin{align*}
K
&=\binom{k}{2},
&
M
&\leq Kq^2=O(k^4).
\end{align*}
Apply \cref{thm:csp-encoder} and use the same rational midpoint threshold as in the proof of \cref{thm:exact-hardness}.  With $m=q/2$, the construction has
\begin{align*}
d
&=2k+1,
&
n
&=2km+M+1=kq+M+1=O(k^4).
\end{align*}
Every coordinate has bit length polynomial in $k,m,K,M$, so the total encoding length is $L=k^{O(1)}=d^{O(1)}$.

Consequently, a $2^{o(d\log d)}\poly(L)$-time algorithm for the constructed instances would solve $k\times k$ Clique in
\begin{align*}
2^{o((2k+1)\log(2k+1))}k^{O(1)}
&=2^{o(k\log k)}
\end{align*}
time, contradicting ETH.  Finally, substituting $n,L=d^{O(1)}$ into \cref{thm:known-xp} gives the stated $2^{O(d\log d)}$ upper bound.
\end{proof}

For $p=2$, clearing a common denominator in the constructed columns yields the following stronger arithmetic specialization.

\begin{corollary}[Euclidean integer specialization]
\label{cor:euclidean-integer}
Euclidean longest vector sum is W[1]-hard in the dimension and has the same ETH lower bound even for integer $5$-sparse generators of polynomial encoding length.
\end{corollary}
\begin{proof}
For $p=2$, multiply every constructed column by a common denominator and multiply the threshold by its square.  This preserves all subset comparisons, the dimension, and sparsity.  The common denominator has polynomial bit length because the sum of the input denominator bit lengths is polynomial.
\end{proof}

\section{The Precision Frontier}
\label{sec:approximation}

We first give a deterministic fixed-$p$ upper bound and then use the inverse-polynomial gap of the cap encoder to prove that the exponent of $1/\varepsilon$ cannot be sublinear in the dimension.

\begin{definition}[Precision-FPTAS]
\label[definition]{def:precision-fptas}
A precision-FPTAS for fixed-$p$ generator-zonotope norm maximization returns a feasible subset with ratio at least $1-\varepsilon$ in time $f_p(d)\poly(L,1/\varepsilon)$, where the polynomial degree is independent of $d$ and $\varepsilon$.
\end{definition}

\begin{lemma}[Ball--Carlen--Lieb smoothness, unit-sphere form \cite{bcl94}]
\label[lemma]{lem:quadratic-support-loss}
Let $r\geq2$ and define $J_r(u)_j\coloneqq\sign(u_j)\abs{u_j}^{r-1}$. For unit vectors $u,v\in\ell_r^d$,
\begin{align*}
0
&\leq 1-J_r(u)^\top v
\leq \frac{r-1}{2}\norm{u-v}_r^2.
\end{align*}
\end{lemma}

The Euclidean fixed-cardinality scheme in \cite[Theorem~5]{bgg+08} already has $\varepsilon^{-(d-1)/2}$ dependence; the unconstrained case follows by scanning the cardinality.  The next theorem extends this dependence to every fixed rational $p\in(1,\infty)$ and records a subset-returning rational implementation.

\begin{theorem}[Deterministic $(1-\varepsilon)$ approximation]
\label{thm:precision-upper}
Fix a rational $p\in(1,\infty)$. Given $A\in\Q^{d\times n}$ and a rational $\varepsilon\in(0,1/2)$, one can deterministically return $\widehat x\in\set{0,1}^n$ such that
\begin{align*}
\norm{A\widehat x}_p
&\geq(1-\varepsilon)R_p(A)
\end{align*}
in
\begin{align*}
f_p(d)\varepsilon^{-(d-1)/2}\poly(L,\log(1/\varepsilon))
\end{align*}
bit operations and polynomial working space.
\end{theorem}
\begin{proof}
Take coprime positive integers $a,b$ with $p=a/b$, define $q\coloneqq p/(p-1)$ and $r\coloneqq\max\set{p,q}\geq2$, and take coprime positive integers $c,e$ with $r=c/e$.  Thus $c=a$ and $e$ is either $b$ or $a-b$.  Cover the unit $\ell_r$ sphere by the $2d$ charts on which one signed coordinate has magnitude at least $d^{-1/r}$.

We use a grid adapted to the denominator $e$.  Choose a dyadic mesh $h$ as the largest power of two at most $\delta/(4ed^2)$.  On the chart with omitted coordinate $\ell$ and sign $\sigma$, enumerate grid values $\theta_j\in[-1,1]\cap h\mathbb Z$ for $j\neq\ell$ satisfying $\sum_{j\neq\ell}\abs{\theta_j}^c\leq1$, and set
\begin{align*}
u_j
&\coloneqq\sign(\theta_j)\abs{\theta_j}^e
&&\text{for }j\neq\ell,
&
u_\ell
&\coloneqq
\sigma\left(1-\sum_{j\neq\ell}\abs{\theta_j}^c\right)^{e/c}.
\end{align*}
Then $\norm{u}_r=1$.  To verify the covering radius, take a target $v$ in this chart, define $\vartheta_j\coloneqq\sign(v_j)\abs{v_j}^{1/e}$, and round each $\vartheta_j$ toward zero to obtain $\theta_j$.  For $j\neq\ell$,
\begin{align*}
\abs{u_j-v_j}
&\leq eh,
&
0\leq\abs{v_j}^r-\abs{u_j}^r
&\leq ch.
\end{align*}
Since $\abs{v_\ell}^r\geq1/d$, the derivative of $s\mapsto s^{1/r}$ on the relevant interval is at most $d^{1-1/r}/r$.  Hence
\begin{align*}
\abs{u_\ell-v_\ell}
&\leq e(d-1)d^{1-1/r}h
\leq ed^2h.
\end{align*}
Thus $\norm{u-v}_r\leq\norm{u-v}_1\leq2ed^2h\leq\delta$, and the net has size at most
\begin{align*}
2d\left(1+\frac{16ed^2}{\delta}\right)^{d-1}.
\end{align*}
This grid also gives the required certified Turing implementation.  The quantity $s\coloneqq1-\sum_{j\neq\ell}\abs{\theta_j}^c$ is rational.  If $s>0$, isolate the nonnegative root $\rho=s^{1/a}$ by bisection of $X^a-s$; if $s=0$, take $\rho=0$.  When $p\geq2$, the non-omitted coordinates of $u$ are signed $b$th powers of rationals, the non-omitted coordinates of $J_p(u)$ are the corresponding signed $(a-b)$th powers, and the two omitted coordinates are $\sigma\rho^b$ and $\sigma\rho^{a-b}$.  When $p<2$, the net is on the dual sphere and its omitted coordinate is $\sigma\rho^{a-b}$.  Thus every needed direction involves only one fixed-degree $a$th root, rather than a root of a sum of rational powers.  Coordinatewise rational bisection therefore produces a rational direction within dual-norm error $\eta$ in time polynomial in $\log d$, $\log(1/\delta)$, and $\log(1/\eta)$ for fixed $p$.

If $A=0$, return the empty subset.  Otherwise some singleton subset has positive norm.  Let $z^*\coloneqq Ax^*$ be optimal and put $R\coloneqq\norm{z^*}_p>0$. If $p\geq2$, net the primal sphere near $v^*\coloneqq z^*/R$ and map a nearby point $u$ to the dual direction $J_p(u)$. If $p<2$, put $y^*\coloneqq J_p(v^*)$ and net the dual sphere near $y^*$. In either case, \cref{lem:quadratic-support-loss} supplies a unit dual direction $y$ such that
\begin{align*}
y^\top v^*
&\geq1-\frac{r-1}{2}\delta^2.
\end{align*}
Let $\widehat y\in\Q^d$ be its rationalization with $\norm{\widehat y-y}_q\leq\eta$. The directional subset
\begin{align*}
S(\widehat y)
&\coloneqq\set{i:\widehat y^\top a_i\geq0}
\end{align*}
maximizes the rational support $\widehat y^\top Ax$ over all subsets. Hence
\begin{align*}
\norm{A\mathbf{1}_{S(\widehat y)}}_p
&\geq \frac{\widehat y^\top A\mathbf{1}_{S(\widehat y)}}{\norm{\widehat y}_q} \\
&\geq R\frac{1-(r-1)\delta^2/2-\eta}{1+\eta}.
\end{align*}
Choose $\delta$ as the largest dyadic with $\delta^2\leq\varepsilon/[2(r-1)]$ and set $\eta\coloneqq\varepsilon/8$. Scanning all rationalized directions and retaining the one with largest rational support returns a $(1-\varepsilon)$ solution. The net size is $f_p(d)\varepsilon^{-(d-1)/2}$, each scan uses polynomially many rational operations, and directions can be generated one at a time.
\end{proof}

Combining \cref{thm:precision-upper} with \cref{prop:turing-xp} gives the better of the precision-dependent bound and the $n^{O(d)}b^{O(d)}\poly(L)$ exact Turing baseline. The real-arithmetic exponent also follows by combining classical polyhedral approximation with polyhedral-norm longest vector sum; the theorem above records the subset-returning rational implementation needed for the common upper/lower model.

\begin{lemma}[Quantitative cap gap]
\label{lem:quantitative-gap}
For a $k$-Multicolored-Clique instance with color classes of size at most $N$, the construction in \cref{thm:csp-encoder} can be instantiated so that
\begin{align*}
\tau
&\geq\frac{1}{c_p k^4N^8}.
\end{align*}
There are rational values $Y>U$ such that a YES instance has $\OPT_p(A)\geq Y$, a NO instance has $\OPT_p(A)\leq U$, and
\begin{align*}
Y-U
&>\frac{4\tau}{5},
&
Y
&\leq k+2.
\end{align*}
\end{lemma}
\begin{proof}
Pad the label domain with isolated dummy labels so that its size is $2m\geq4$ with $m\leq N$ up to a constant factor.  Choose $\Gamma=\Theta_p(m^{-2})$, and enlarge the uniform cap bound if necessary so that
\begin{align*}
1\leq C_0
&\leq c_p' N^2.
\end{align*}
There is one cap per accepted edge tuple and pair constraint, so $M\leq\binom{k}{2}N^2\leq k^2N^2$.  Choose $\tau$ as the largest dyadic satisfying the upper bounds in the proof of \cref{thm:csp-encoder}.  Substituting $d=2k+1$, the bounds on $M$, and the bound on $C_0$ shows that every term in that minimum is at least $1/(c_p k^4N^8)$.  Maximal dyadic rounding changes only the constant.

Let
\begin{align*}
Y
&\coloneqq k+1-k\eta+\tau K(1-\xi)-\frac{\tau}{20}, \\
U
&\coloneqq k+1+k\eta+\tau(K-1)(1+\xi)+\frac{\tau}{20}.
\end{align*}
The interval separation in \cref{thm:csp-encoder} gives $Y-U>4\tau/5$.  On a YES instance $M\geq K$, and the parameter choice gives $\tau K\leq\tau MC_0\leq1/16$.  Hence $Y\leq k+2$.
\end{proof}

\begin{theorem}[Precision lower bounds]
\label{thm:precision-lower}
Fix a rational $p\in(1,\infty)$.  Unless FPT$=$W[1], no precision-FPTAS exists, even for $5$-sparse generators.  Under ETH, there is no algorithm with running time
\begin{align*}
\rho_p(d)(1/\varepsilon)^{h(d)}L^c
\end{align*}
for any computable $\rho_p$, constant $c$, and function $h(d)=o(d)$.
\end{theorem}
\begin{proof}
Apply \cref{lem:quantitative-gap} and request accuracy
\begin{align*}
\varepsilon_{\mathrm{app}}
&\coloneqq\frac{\tau}{10p(k+2)}.
\end{align*}
On a YES instance, a returned $(1-\varepsilon_{\mathrm{app}})$ solution $\widehat z$ satisfies, by Bernoulli's inequality and $Y\leq k+2$,
\begin{align*}
F_p(\widehat z)
&\geq(1-\varepsilon_{\mathrm{app}})^pY \\
&\geq Y-p\varepsilon_{\mathrm{app}}Y \\
&\geq Y-\frac{\tau}{10}
>U+\frac{7\tau}{10}.
\end{align*}
On a NO instance every feasible subset has power at most $U$.  The returned subset sum is rational, and its coordinate powers can be enclosed to total additive error below $\tau/100$ by fixed-degree bisection.  This decides the source instance in polynomial postprocessing time.

By \cref{lem:quantitative-gap},
\begin{align*}
\varepsilon_{\mathrm{app}}^{-1}
&=O_p(k^5N^8).
\end{align*}
A precision-FPTAS would therefore solve Multicolored Clique in $f_p(k)N^{O(1)}$ time, contradicting W[1]-hardness.  For the ETH statement, an algorithm with exponent $h(d)=o(d)$ and $d=2k+1$ would run in
\begin{align*}
\rho_p(2k+1)\left(O_p(k^5N^8)\right)^{h(2k+1)}\poly(k,N)
&=\rho_p'(k)N^{o(k)},
\end{align*}
contradicting \cref{thm:clique-lower}.
\end{proof}

Among algorithms with fixed-degree polynomial dependence on $L$, \cref{thm:precision-upper,thm:precision-lower} give an upper exponent $(d-1)/2$ for $1/\varepsilon$, while ETH excludes every $o(d)$ exponent.

\section{Positive-Output ReLU Consequence}

\begin{proposition}[ReLU--zonotope correspondence, Section~2 of \cite{fgh+25}]
\label{prop:relu-duality}
Let $p,q\in(1,\infty)$ be dual exponents, let $A\coloneqq(a_1,\ldots,a_n)\in\R^{d\times n}$, and define $f_A(u)\coloneqq\sum_{i=1}^n\ReLU(a_i^\top u)$.  Then
\begin{align*}
\Lip_q(f_A)
&=\max_{z\in Z(A)}\norm{z}_p
=R_p(A).
\end{align*}
\end{proposition}

\begin{corollary}[Positive-output ReLU Lipschitz hardness]
\label{cor:relu-hardness}
Fix a rational $q\in(1,\infty)$ and let $p\coloneqq q/(q-1)$.  Given $A\in\Q^{d\times n}$ and $B\in\Q_{>0}$, deciding whether $\Lip_q(f_A)^p\geq B$ is W[1]-hard in $d$ and admits no $\rho_q(d)L^{o(d)}$ algorithm under ETH, even when all output weights are $+1$, all biases are zero, and all hidden weights are $5$-sparse.  The precision bounds of \cref{thm:precision-upper,thm:precision-lower} transfer with the same dependence on $d$ and $\varepsilon$.
\end{corollary}
\begin{proof}
Apply \cref{prop:relu-duality} to \cref{thm:exact-hardness,thm:precision-upper,thm:precision-lower}; each generator becomes one hidden unit, so all stated restrictions are preserved.
\end{proof}

For $p=2$, the identity $x^\top A^\top Ax=\norm{Ax}_2^2$ also gives W[1]-hardness and the same ETH lower bound for factorized positive-semidefinite binary quadratic maximization parameterized by the supplied factor rank. Both sides of the precision frontier transfer under the same identity.

\section{Related Work}

\paragraph{Norm maximization and representation.}
The parameterized complexity of fixed-$\ell_p$ norm maximization over H-presented symmetric polytopes was classified in \cite{kkw15}.  Its hardness construction intersects a product of curved planar scaffolds with symmetric strips that exclude forbidden label pairs.  Our construction replaces each graph-dependent cut by a positive cap generator.  This distinction is necessary: intersecting $[-1,1]^3$ with $\abs{x_1+x_2+x_3}\leq2$ creates a triangular face on $x_1+x_2+x_3=2$, and hence does not produce a zonotope.  Thus the H-presented reduction does not preserve the generator model.

The generator problem also appears as longest vector sum, low-rank convex QUBO, and induced matrix-norm computation.  Exact algorithms enumerate arrangement cells or zonotope vertices in $n^{O(d)}$ time \cite{ffl05,kl10,mkp14,she20}.  The generator representation was isolated as an open parameterized case in \cite{fgh+25,fgh+26}, including through its equivalence to positive-output two-layer ReLU Lipschitz computation.

\paragraph{Approximation and containment.}
For scaled zonotope containment, \cite{err+26} gives a randomized $O(\sqrt d)$ approximation, a nearly matching $\Omega(\sqrt{d/\log d})$ oracle lower bound, and sparsification results.  Its objective $\max\set{s>0:sZ\subseteq Q}$ with oracle-given $Q$ differs from fixed-$\ell_p$ maximization over $Z$.

For approximation, \cite{she18} gives randomized fixed-accuracy guarantees under an arbitrary evaluable norm.  The order reduction in \cite{fgh+26}, based on Lewis-weight embeddings \cite{cp15,brr23}, gives another randomized FPT approximation followed by vertex enumeration.  For explicitly polyhedral norms, \cite{she20} optimizes one directional subset per facet.

Classical approximation of smooth convex bodies by polytopes gives the $\varepsilon^{-(d-1)/2}$ facet scale \cite{bi75,adm24}. Combined with \cite{she20}, it yields the real-arithmetic exponent in \cref{thm:precision-upper}; our formulation gives a rational, subset-returning implementation for fixed $\ell_p$.

The approximation lower bound in \cite{kkw15} applies to H-presented polytopes. For generator zonotopes, \cref{thm:precision-lower} obtains the corresponding obstruction directly from the quantitative gap in our cap construction.

\PaperPrintBibliography

\PaperEndMatter

\end{document}